\documentclass{article}
\usepackage[utf8]{inputenc} 
\usepackage[T1]{fontenc}
\usepackage{url}              
\usepackage{cite}             

\usepackage[cmex10]{amsmath}  
\usepackage{mleftright}       
\mleftright                   

\usepackage{graphicx}         
\usepackage{booktabs}         
\usepackage{mathdots}

\usepackage[utf8]{inputenc} 
\usepackage[T1]{fontenc}
\usepackage{url}
\usepackage{ifthen}
\usepackage{cite}
\usepackage{algorithm}
\usepackage{algpseudocode}
\usepackage[cmex10]{amsmath} 
\usepackage{amssymb}
\usepackage{xcolor}

\usepackage{subcaption}
\usepackage{comment}
\usepackage{tabularx}
\usepackage{mathtools}

\usepackage{tikz}
\usetikzlibrary {arrows.meta,automata,positioning} 

\usepackage{amsfonts, amsthm}

\usepackage{graphicx}%
\usepackage{multirow}%
\usepackage{amsmath,amssymb,amsfonts}%
\usepackage{amsthm}%
\usepackage{mathrsfs}%
\usepackage[title]{appendix}%
\usepackage{xcolor}%
\usepackage{textcomp}%
\usepackage{manyfoot}%
\usepackage{booktabs}%
\usepackage{algorithm}%
\usepackage{algorithmicx}%
\usepackage{algpseudocode}%
\usepackage{listings}%

\usepackage{hyperref}
\usepackage{mathtools}
\usepackage{enumerate}
\usepackage{tabularx}
\usepackage{makecell}

\usepackage{tikz}
\usetikzlibrary{matrix}
\usepackage{ulem}
\usepackage{caption}
\usepackage{lipsum}
\usepackage{parskip}
\usepackage{resizegather}

\usepackage{geometry}

\makeatletter
\newcommand{\multiline}[1]{%
  \begin{tabularx}{\dimexpr\linewidth-\ALG@thistlm}[t]{@{}X@{}}
    #1
  \end{tabularx}
}
\makeatother

\newtheorem{theorem}{Theorem}

\newtheorem{lemma}{Lemma}

\theoremstyle{definition}
\newtheorem{definition}{Definition}
\theoremstyle{remark}
\newtheorem{remark}{Remark}

\usepackage{graphicx}%
\usepackage{multirow}%
\usepackage{amsmath,amssymb,amsfonts}%
\usepackage{amsthm}%
\usepackage{mathrsfs}%
\usepackage[title]{appendix}%
\usepackage{xcolor}%
\usepackage{textcomp}%
\usepackage{manyfoot}%
\usepackage{booktabs}%
\usepackage{algorithm}%
\usepackage{algorithmicx}%
\usepackage{algpseudocode}%
\usepackage{listings}%
\usepackage{cleveref}
\usepackage{enumitem} 
\usepackage{anyfontsize}

\title{Construction of LDPC convolutional codes with large girth from Latin squares}
\author{Elisa Junghans and Julia Lieb}
\date{}

\begin{document}

\maketitle

\abstract{Due to their capacity approaching performance low-density parity-check (LDPC) codes gained a lot of attention in the last years. The parity-check matrix of the codes can be associated with a bipartite graph, called Tanner graph. To decrease the probability of decoding failure it is desirable to have LDPC codes with large girth of the associated Tanner graph. Moreover, to store such codes efficiently, it is desirable to have compact constructions for them. In this paper, we present constructions of LDPC convolutional codes with girth up to $12$ using a special class of Latin squares and several lifting steps, which enables a compact representation of these codes.
With these techniques, we can provide constructions for well-performing and efficiently storable time-varying and time-invariant LDPC convolutional codes as well as for LDPC block codes.
}

\section{Introduction}
Low-density parity-check (LDPC) codes were first introduced by Gallager in 1962 \cite{gallager1962low}.
In recent years, these codes gained a lot of interest because of their capacity approaching performance with message passing algorithms together with their low encoding and decoding complexity.
LDPC codes are characterized by the property of possessing a sparse parity-check matrix and can be described via a bipartite graph, called Tanner graph \cite{tanner1981recursive}. These properties can be generalized to the setting of convolutional codes, both (periodically) time-varying and time-invariant, to obtain LDPC convolutional codes. These codes are also known as spatially coupled LDPC codes and were introduced by Jimenez-Felstrom and Zigangirov in 1999 \cite{felstrom1999time}.
LDPC convolutional codes have been shown to be capable of achieving the same capacity-approaching performance as LDPC block codes with message passing decoding algorithms.
For these decoding algorithms to perform well, for block codes as well as convolutional codes, it is desirable to maximize the girth, i.e. the length of the shortest cycle, of the associated Tanner graph, see e.g. \cite{harmful_objects}.

While it is possible to find well-performing LDPC codes via random search, it is still desirable to construct such codes that additionally allow for some kind of compact representation in order to store them efficiently.
For this reason, there is a huge amount of papers on quasi-cyclic LDPC codes. Moreover, several papers use combinatorial constructions to achieve a compact representation, see e.g. \cite{xie2016euclidean,zhang2022spatially,zhang2016time}.

Another construction technique for LDPC codes that is frequently used, mainly to remove harmful cycles in the Tanner graph, is constructing first a so-called protograph
and then applying some lifting procedure to expand this graph, see e.g.
\cite{karimi2013girth}\cite{smarandache2022unifying}\cite{fang2015survey}\cite{mitchell2014quasi}.
Similar techniques have also been exploited for the construction of LDPC convolutional codes, see e.g. \cite{scprotograph}, \cite{qcliftingprotograph}, \cite{ccprotograph}.

There is a large variety of papers showing the excellent performance of LDPC convolutional codes, see e.g.
\cite{pseudocodeword}\cite{implementation}\cite{implementation2}\cite{decoding}\cite{decoding2}.
Often LDPC convolutional codes are constructed from LDPC block codes via so-called unwrapping techniques. For most of these constructions the obtained LDPC convolutional codes outperform the LDPC block codes they were constructed from, see
\cite{survey}\cite{convolutionalgain}\cite{QC}\cite{circulant}\cite{rate1/2fromblock}\cite{proto_unwrapping_performance}\cite{qcunwrapping}.
As mentioned above, it is important to maximize the girth of LDPC (convolutional) codes. In \cite{girthproperties}, upper bounds for the girth for certain types of LDPC convolutional codes are presented but concrete examples of codes are only obtained via computer search. In \cite{unwrapping_maxgirth}, time-varying LDPC convolutional codes with large girth are constructed from LDPC block codes with large girth.
Explicit constructions for LDPC convolutional codes can be found e.g. in \cite{monomial}, \cite{alfarano2021construction} or \cite{highrate_lifting}. The last of these papers only considers high-rate codes but also uses some lifting technique starting from circulant matrices.

In this paper, we present a construction for periodically time-varying LDPC convolutional codes starting from a special class of orthogonal Latin squares. To achieve a larger girth, we apply several lifting steps to the original construction. The definition of these codes via concrete Latin squares and well-determined lifting steps allows for a very compact representation of these codes.
We use similar techniques to also construct time-invariant LDPC convolutional codes of large girth and with a compact representation. Moreover, we use our techniques to increase the girth of the LDPC block code construction from \cite{BlockCodes}, which also uses Latin squares to define the parity-check matrices of the codes.

The paper is structured as follows.
In Section 2, we provide the definitions and basics on convolutional codes, LDPC codes and Latin squares that we will need in the following parts of the paper. In Section 3, we present our main results, i.e. the construction of periodically time-varying convolutional codes of girth up to 12 using a special class of Latin squares and several lifting steps. In Section 4, we use similar techniques to construct time-invariant LDPC convolutional codes of large girths. In Section 5, we use special lifting steps to remove all 6 cycles from the LDPC block code construction based on Latin squares from \cite{BlockCodes}.

\section{Preliminaries}
\subsection{Convolutional codes}\label{subsec CC}

In this subsection, we introduce time-invariant and time-varying convolutional codes.
These codes can be defined over any finite field, however, in this paper we only consider binary convolutional codes, i.e. codes over the finite field with $2$ elements, denoted by $\mathbb{F}_2$. Furthermore, we denote by $\mathbb F_2[z]$ the polynomial ring over $\mathbb F_2$.
  
 \begin{definition}
    An $(n,k)$ binary (time-invariant) \textbf{convolutional code} $\cal C$ is defined as an $\mathbb F_2[z]$-submodule of $\mathbb{F}_2[z]^n$ of rank $k$. Hence, there exists a polynomial \textbf{generator matrix} $G(z) \in \mathbb F_2[z]^{k \times n}$ whose rows form a basis of $\cal C$, i.e.,
    $$
{\cal C}= 
\{v(z) \in \mathbb{F}_2[z]^n \mid v(z)=u(z)G(z) \mbox{ for some } u(z) \in \mathbb{F}_2[z]^k \}.
    $$
  \end{definition}

The generator matrix of a convolutional code is not unique and two polynomial matrices
$G(z),\tilde G(z) \in \mathbb F_2[z]^{k \times n}$ are generator matrices of the same convolutional code if and only if $\tilde G(z)=U(z)G(z)$ for 
$U(z) \in \mathbb F_2[z]^{k \times k}$ which has an inverse over $\mathbb F_2[z]$.

If any generator matrix of a convolutional code $\mathcal{C}$ is \textbf{leftprime}, i.e., has a polynomial right inverse, then the same is true for all generator matrices of $\mathcal{C}$.
In this case, there exists a full row-rank \textbf{parity-check matrix} $H(z)\in\mathbb F_2[z]^{(n-k)\times n}$ 
such that
$$\mathcal C := \{v(z)\in\mathbb F_2[z]^n \mid H(z)v(z)^\top={0}\}.$$

We write $H(z)=\sum_{i=0}^{\mu}H_iz^i$ with $H_i\in\mathbb{F}_2^{(n-k)\times n}$ and $H_{\mu}\neq 0$ and define $\deg(H(z))=\mu$.

With this notation, we can expand the kernel representation $H(z)v(z)^\top=0$ for $s = \deg(v)$ and $v(z)=\sum_{i=0}^s v_iz^i$ where $v_i\in \mathbb{F}_2^n$ in the following way:
\begin{equation}\label{eq:kerH}
H_sv^\top := \begin{bmatrix}
H_0 & & & &  \\
\vdots & \ddots & & &  \\
H_{\mu} & \cdots & H_0 & &  \\
 & \ddots &  &\ddots &  \\
 & & H_{\mu} &\cdots & H_0 \\
 & & & \ddots & \vdots \\
 & & & & H_\mu
\end{bmatrix}\begin{bmatrix}
v_0 \\ v_1 \\ \vdots \\ v_s
\end{bmatrix}=0.
\end{equation}

The matrix $H_s$ in the above equation is called $s$-th
 \textbf{sliding parity-check matrix}. 

Using the representation of a (time-invariant) convolutional code given in \eqref{eq:kerH}, we next present an analogue definition for time-varying convolutional codes.

 Let $n,k,\mu,s\in \mathbb{N}$ with $k<n$ and $H_j(t)\in \mathbb{F}_2^{(n-k)\times n}$ for $j\in\{0,\dots,\mu\}$ and $t\in\mathbb{N}_0$.
    A time-varying \textbf{sliding parity-check matrix} $H_{[0,s]}$ is defined through
    \begin{align*}
        H_{[0,s]} :=\begin{pmatrix}
            H_0(0)\\
            H_1(0)& H_0(1)\\
            \vdots& H_1(1)& \ddots &\\
            H_\mu (0) &\vdots &\ddots &H_0(s)\\
            & H_\mu(1) & & H_1(s)\\
            &&\ddots& \vdots\\
            &&& H_\mu(s)
        \end{pmatrix}\in \mathbb{F}_2^{(\mu+s+1)(n-k)\times (s+1)n}.
    \end{align*}
   
\begin{definition}
For $v(z)=\sum_{i=0}^s v_iz^i\in\mathbb F_2[z]^n$ define $v:=(v_0,\hdots,v_s)$ with $s=\deg(v)$.
    A time-varying $(n,k)$ \textbf{convolutional code} $\mathcal{C}$ is defined as 
    $$\mathcal{C}:=\{v(z)\in\mathbb F_2[z]^n \mid H_{[0,s]}v^\top=0 \text{ for } s=\deg(v)\}.$$
    A time-varying $(n,k)$ convolutional code $\mathcal{C}$ has \textbf{period} $T$ if $H_j(t)=H_j(t+T)$ for all $j\in\{0,\dots,\mu\}$ and $t\in\mathbb{N}_0$ and if $T=1$, we obtain a time-invariant convolutional code.
\end{definition}    

Since time-varying convolutional codes are a generalization of time-invariant convolutional codes, we present the following definitions and results only for the general case of time-varying convolutional codes.

For convolutional codes, there exist different distance notions to measure different aspects of the error-correcting capability of the code.
In this paper, we will consider the free distance, which measures how many errors can be corrected throughout the whole codeword, and the $j$-th column distances, which is a measure for how many errors can be corrected sequentially with time-delay (at most) $j$, assuming that $v_i$ (containing possible errors) is received at time-instant $i$.

The (Hamming) weight of a vector $v(z)=\displaystyle \sum_{i \in \mathbb N_0} v_i z^i \in \mathbb{F}_2[z]^n$ is defined as
$$
{\rm wt}(v(z))= \sum_{i \in \mathbb N_0} {\rm wt}(v_i)
$$
where ${\rm wt}(v_i)$ is the number of nonzero entries of $v_i$. 

\begin{definition}
    Let $\cal C$ be an $(n,k)$ time-varying convolutional code. The \textbf{free distance} of $\cal C$ is
    $$
    d_{free}({\cal C}):=min\{{\rm wt}(v(z)) \mid v(z) \in {\cal C}\backslash \{0\}\}.
    $$
     For any $j\in\mathbb{N}_0$ we define the \textbf{$j$-th column distance} of $\mathcal{C}$ as 
    \begin{align*}
        d_j^c(\mathcal{C}) &:= \min \{ wt((v_0,\dots,v_j)) \mid H_j^c(t) (v_0,\dots,v_j)^T=0, v_0\neq 0, t\in \mathbb{N}_0 \}
    \intertext{for the \textbf{$j$-th truncated sliding parity-check matrix at time $t\in \mathbb{N}_0$}}
        H_j^c(t)&:= \begin{pmatrix}
            H_0(t)\\
            H_1(t) & H_0(t+1)\\
            \vdots & \vdots &\ddots\\
            H_j(t) & H_{j-1}(t+1) &\cdots &H_0(t+j)
        \end{pmatrix}\in \mathbb{F}_2^{(j+1)(n-k)\times (j+1)n}.
    \end{align*}
\end{definition}

If $\mathcal{C}$ is time-invariant, i.e. for all $j\in\mathbb N_0$, $H^c_j(t)=H^c_j(t^\prime)$ for all $t,t^\prime \in\mathbb N_0$, this definition of the $j$-th column distance simplifies to
\begin{align*}
        d_j^c(\mathcal{C}) &= \min \{ wt((v_0,\dots,v_j)) \mid H_j^c (v_0,\dots,v_j)^T=0, v_0\neq 0 \}
    \end{align*}
where $H_j^c:=H_j^c(t)$.

To calculate the $j$-th column distance of a time-varying convolutional code, we have to calculate for each fixed $t\in\mathbb N_0$, the column distance $d_{j,t}^c$ of the time-invariant convolutional code with sliding parity-check matrix $H_j^c(t)$ and then take the minimum of all these values.

The following theorem shows how to calculate the column distances of a time-invariant convolutional code using its parity-check matrix.
    
    \begin{theorem}\label{Th djc Hjc}
    \cite[Proposition 2.1]{gluesing2006strongly}
           Let $\mathcal{C}$ be a time-invariant $(n,k)$ convolutional code and $d\in\mathbb{N}$. Then, the following properties are equivalent:
        \begin{enumerate}[label=(\roman*)]
            \item  $d_j^c(\mathcal{C})=d$
            \item none of the first $n$ columns of $H_j^c$ is contained in the span of any other $d-2$ columns and one of the first $n$ columns of $H_j^c$ is in the span of some other $d-1$ columns of that matrix.
        \end{enumerate}
    \end{theorem}

From this theorem we can immediately deduce the corresponding statement for time-varying convolutional codes.

   \begin{theorem}\label{Th djc Hjc time}
           Let $\mathcal{C}$ be a time-varying $(n,k)$ convolutional code and $d\in\mathbb{N}$. Then, the following properties are equivalent:
        \begin{enumerate}[label=(\roman*)]
            \item $d_j^c(\mathcal{C})=d$
            \item there exists $t\in\mathbb N_0$ such that $d_{j,t}^c=d$ and $d_{j,t}^c\geq d$ for all $t\in\mathbb N_0$.
            \item  there exists $t\in\mathbb N_0$ such that one of the first $n$ columns of $H_j^c(t)$ is in the span of some other $d-1$ columns of that matrix and for all $t\in\mathbb{N}_0$ none of the first $n$ columns of $H_j^c(t)$ is contained in the span of any other $d-2$ columns of this matrix.
        \end{enumerate}
    \end{theorem}

\subsection{LDPC codes}\label{subsec LDPC}

A binary \textbf{LDPC (Low-density parity-check) code} $\mathcal{C}$ is defined as the kernel of a sparse parity-check matrix $H\in\mathbb{F}_2^{N\times M}$.
A convolutional code is called LDPC if the associated sliding parity-check matrix is sparse. 
One can associate such a parity-check matrix with a bipartite graph called the \textbf{Tanner graph}, where the set of vertices consists of $M$ independent \textbf{variable nodes} $\{v_1,\dots,v_M\}$ and $N$ independent \textbf{check nodes} $\{w_1,\dots,w_N\}$ and a variable node $v_m$ is adjacent to a check node $w_n$ if and only if the $(n,m)$-entry of $H$ is equal to $1$.

A cycle in the Tanner graph always has an even length. The length of a shortest cycle is called the \textbf{girth} of the Tanner graph, the girth of $H$ or the girth of $\mathcal{C}$. A bigger girth provides less decoding failure (see e.g. \cite{harmful_objects}), thus we want to construct codes with large girth.
A cycle of length $2\ell$ in the Tanner graph is represented by a submatrix of $H$ of the form
\begin{align}
    \begin{pmatrix}
        1&1&0&\cdots &0\\
        0&1&1& &\vdots\\
        \vdots & & \ddots&\ddots&\vdots\\
        0&&& 1&1\\
        1&0&\cdots &0&1
    \end{pmatrix}\in \mathbb{F}_2^{\ell\times\ell} \label{2lcycle matrix}
\end{align}
up to row and column permutations.
Therefore, our aim is to construct parity-check matrices without such submatrices.

\subsection{Latin squares}\label{subsec LS}

In this subsection, we introduce Latin squares, which we will use for our construction of LDPC codes with large girth.  
  
  \begin{definition}
      A \textbf{Latin square} of order $p$ is a function $L:\{1,\dots,p\}\times\{1,\dots,p\}\rightarrow A$ with $\vert A\vert =p$ such that $L(i,j)=L(i^\prime,j)$ implies $i=i^\prime$ and $L(i,j)=L(i,j^\prime)$ implies $j=j^\prime$. We can write $L$ as a $p\times p$ matrix where each column and row contain each entry exactly once. \\
    A pair of Latin squares $L_1$ and $L_2$ is called \textbf{orthogonal} if for every $(i_1,i_2)\in A^2$ there is exactly one pair $(a,b)$ such that $L_1(a,b)=i_1$ and $L_2(a,b)=i_2$.
  \end{definition}
  
    For our construction of LDPC convolutional codes, we use the following specific construction of pairwise orthogonal Latin squares.

 \begin{theorem}\cite{Lidl_Niederreiter_1996}
         Let $p$ be a prime number, and let $A$ be the set $\{1,\dots,p\}$. Consider the matrices $L_1,\dots,L_{p-1}$ of order $p$ defined by
    \begin{align}
        L_r (a,b) := b-r(a-1) \mod p && r=1,\dots,p-1;\ a,b=1,\dots,p \label{Def LQ}
    \end{align}
    where we use $p$ instead of $0$. Then, $L_1,\dots,L_{p-1}$ form a set of pairwise orthogonal Latin squares.
 \end{theorem}

    For $r=1,\dots,p-1$ we define for the Latin square $L_r$ the incidence matrices $Q^r_1,\dots,Q^r_p\in \mathbb{F}_2^{p\times p} $ through
    \begin{align}
        Q_i^r(a,b) := \left\{\begin{array}{ll}1 & \text{if } L_r(a,b)=i \\ 0 & \, \text{otherwise} \\ \end{array} \right. .\label{Def Q}
    \end{align}  
    These matrices are permutation matrices because in a Latin square each entry appears in each row and column exactly once.

\section{Construction of time-varying LDPC convolutional codes with large girth}

  In this section, we present our main results, i.e. we present different constructions for time-varying LDPC convolutional codes with large girth using the Latin squares from the previous subsection.
  In a first step, we construct the parity-check matrix of such a code with girth $6$. Then, we apply several lifting steps to this parity-check matrix to obtain parity-check matrices of codes with girths $8,10$ and $12$.
  At the end of the section, we study the density and distance properties of the constructed parity-check matrices and codes, respectively.
    
\subsection{Construction of LDPC convolutional codes with girth 6}\label{subsec girth 6}

    For a fixed prime number $p$, using definitions (\ref{Def LQ}) and (\ref{Def Q}), let $H_0(t) := \begin{pmatrix} Q_1^{t \mod (p-1) +1} &\vline & I_p \end{pmatrix}$ and $H_i(t) := \begin{pmatrix} Q_{i+1}^{t \mod (p-1) +1} &\vline & 0_{p\times p} \end{pmatrix}$ for all $t\in\mathbb{N}_0$ and $i=1,\dots,\mu\leq p-2$.
    This results in the time-varying sliding parity-check matrix 
    \begin{align}
        H^0_{[0,s]} &:=\left(\begin{array}{cccccccccccc}
            Q_1^1 & I_p\\
            Q_2^1 & 0 &Q_1^2 &I_p\\
            \vdots&\vdots& Q_2^2&0& \ddots &\\
            Q_{\mu+1}^1&0 &\vdots&\vdots &\ddots &Q_1^{p-1}&I_p \\
            && Q_{\mu+1}^2&0 & & Q_2^{p-1}&0&Q_1^1&I_p\\
            &&&&\ddots& \vdots&\vdots&Q_2^1&0&\ddots\\
            &&&&& Q_{\mu+1}^{p-1}&0&\vdots&\vdots&\ddots & Q_1^R&I_p\\
            &&&&&&&Q_{\mu+1}^1&0& & Q_2^R&0\\
            &&&&&&&&&\ddots &\vdots&\vdots \\
            &&&&&&&&&&Q_{\mu+1}^R&0 
        \end{array}\right)\notag
    \end{align}
    with $R:= s\mod(p-1)+1$ and $ \mu\leq p-2$.
    The code corresponding to $H_{[0,s]}^0$ is a binary, time-varying $(2p,p)$ convolutional code with period $p-1$. We denote this code by $\mathcal{C}^0$.
    \begin{theorem}\label{theo H0 girth 6}
        The time-varying LDPC convolutional code $\mathcal{C}^0$ has girth at least $6$. 
    \end{theorem}
    \begin{proof}
        Assume that there is a cycle of length $4$. Then there exists a submatrix of $H_{[0,s]}^0$ of the form $\big(\begin{smallmatrix}
            1&1\\1&1
        \end{smallmatrix}\big)$. Since columns corresponding to the identity matrices in $H_{[0,s]}^0$ only have one entry equal to $1$, the $1$s forming the $4$-cycle belong to incidence matrices of Latin squares. Thus, there exist $i_1,i_2,r_1,r_2\in \{1,\dots,p-1\}$ and $ a_1,a_2,b_1,b_2\in \{1,\dots,p\}$ and $\alpha\in \mathbb{Z}$ with 
        $\alpha\neq 0$ such that \begin{align}
            \begin{pmatrix}
                Q_{i_1}^{r_1}(a_1,b_1) &Q_{i_2}^{r_2}(a_1,b_2)\\ Q_{i_1+\alpha}^{r_1}(a_2,b_1) & Q_{i_2+\alpha}^{r_2}(a_2,b_2)
            \end{pmatrix} = \begin{pmatrix}
                1&1\\1&1
            \end{pmatrix}\notag \\
            \intertext{which means}
            Q_{i_1}^{r_1}(a_1,b_1) = Q_{i_2}^{r_2}(a_1,b_2) = Q_{i_1+\alpha}^{r_1}(a_2,b_1) = Q_{i_2+\alpha}^{r_2}(a_2,b_2) = 1\notag .
        \end{align}
        Since a Latin square cannot contain two different entries at the same position, we get $a_1\neq a_2$. We also know that $r_1\neq r_2$ because no two incidence matrices of the same Latin square are in the same block row.
        This means that\\
        \begin{minipage}{0.44\textwidth}
            \begin{align}
                i_1 &= b_1-r_1(a_1-1)\mod p \label{4kreis 1}\\
                i_2&= b_2-r_2(a_1-1)\mod p \label{4kreis 2}
            \end{align}
        \end{minipage}
        \begin{minipage}{0.44\textwidth}
            \begin{align}
                i_1+\alpha &= b_1-r_1(a_2-1)\mod p\label{4kreis 3}\\
                i_2+\alpha &= b_2-r_2(a_2-1)\mod p\label{4kreis 4}
            \end{align}
        \end{minipage}
        \begin{align}
            \implies (\ref{4kreis 1})-(\ref{4kreis 2})-(\ref{4kreis 3})+(\ref{4kreis 4}):
            0&= (r_1-r_2)(a_2-a_1)\mod p, \notag
        \end{align}
        which is a contradiction because $r_1\neq r_2$ and $a_1\neq a_2$.
    \end{proof}
\begin{remark}
    We conjecture that the code $\mathcal{C}^0$ from Theorem \ref{theo H0 girth 6} has girth exactly $6$ for $p\geq 5$ and $\mu\geq 3$.
    For $p=5$, we know that this is true, since the $1$s at the positions $Q_3^1(1,3)$, $Q_1^3(1,1)$, $Q_3^3(2,1)$, $Q_4^2(2,1)$,$ Q_3^2(5,1)$ and $Q_4^1(5,3)$ form a cycle of length $6$.
\end{remark}
    
\subsection{Construction of LDPC convolutional codes with girth 8}\label{subsec girth 8}

    We are now going to modify the construction from Subsection \ref{subsec girth 6} and study how the cycles in such a parity-check matrix are located to obtain a code with girth $8$.
    \begin{definition}
        For $m\geq 1$, let $H_{[0,s]}^m$ be the matrix that is constructed from $H_{[0,s]}^{m-1}$ in the following way. Replace each entry of $H_{[0,s]}^{m-1}$ with a matrix of size $p\times p$. Each $0$ is replaced with $0_{p\times p}$, each $1$ that is in $I_p$ is replaced with $I_p$, and each $1$ that is located at position $(a,b)$ in $Q_i^r$ is replaced with $Q_a^r$.
        Furthermore, let $\mathcal{C}^m$ be the binary, time-varying $(2p^{m+1},p^{m+1})$ convolutional code with period $p-1$ that is derived from $H_{[0,s]}^m$.
    \end{definition}
    Each cycle of length $2\ell$ in the code $\mathcal{C}^m$ is located on $\ell$ (not necessarily different) block columns of $H_{[0,s]}^m$. Each of these columns contains only matrices that correspond to the same Latin square because there is only one $1$ in the columns that contain $I_p$ and $0$ everywhere else. We call these Latin squares $L_{r_1},\dots,L_{r_\ell}$ and say that $r_1,\dots,r_\ell$ correspond to the cycle of length $2\ell$. Since there is only one $1$ per row that corresponds to a certain Latin square, 
    each cycle of length $2\ell$ corresponds to $r_1,\dots,r_\ell\in \{1,\dots,p-1\}$ such that $r_j\neq r_{j+1},j=1,\dots,\ell-1$ and $r_1\neq r_\ell$.
Moreover, we will need the following lemma.

\begin{lemma}\label{lemma r pw gleich}
  If there exists a cycle of length $2\ell\leq 2(m+2)$ in $H_{[0,s]}^m$, then for each $i\in\{1,\dots,\ell\}$, there is $j\in\{1,\dots,\ell\}\setminus\{i\}$  such that $r_i=r_j$.
\end{lemma}
    
\begin{proof}
    Assume that there is a cycle of length $2\ell\leq 2(m+2)$ in $H_{[0,s]}^m$.
    Then there is also a cycle of length $2\ell$ in $H_{[0,s]}^0,H_{[0,s]}^1,\dots,H_{[0,s]}^{m-1}$. Such a cycle corresponds to an $\ell\times\ell$ matrix of the form (\ref{2lcycle matrix}). This means for $H_{[0,s]}^0$ that there exist $r_1,\dots,r_{\ell}\in \{1,\dots,p-1\}$ with $r_j\neq r_{j+1},j=1,\dots,\ell-1$ and $r_1\neq r_\ell$ and $a_1^0,\dots,a_\ell^0,a_1^1,\dots,a_\ell^1,b_1^0,\dots,b_\ell^0,b_1^1,\dots,b_\ell^1\in\{1,\dots,p\}$ as well as $\alpha_1,\dots,\alpha_{\ell-1}\in\mathbb{Z}$ such that the submatrix 
    \begin{align*}
        \left(\begin{smallmatrix}
            Q_{a_1^0}^{r_1}(a_1^1,b_1^1) & Q_{a_2^0}^{r_2}(a_1^1,b_2^1)&0&\cdots&0\\
            0& Q_{a_2^0+\alpha_1}^{r_2}(a_2^1,b_2^1)&\hspace{4mm} Q_{a_3^0}^{r_3}(a_2^1,b_3^1) &\ddots&\vdots\\
            \vdots&&\ddots&\ddots&0 \\
            0&&&Q_{a_{\ell-1}^0+\alpha_{\ell-2}}^{r_{\ell-1}}(a_{\ell-1}^1,b_{\ell-1}^1) &Q_{a_\ell^0}^{r_\ell}(a_{\ell-1}^1,b_\ell^1)\\
            Q_{a_1^0+\alpha_1+\dots,+\alpha_{\ell-1}}^{r_1}(a_\ell^1,b_1^1)&0&\cdots &0&Q_{a_\ell^0+\alpha_{\ell-1}}^{r_\ell}(a_\ell^1,b_\ell^1)
        \end{smallmatrix}\right)
    \end{align*}
    is equal to the submatrix (\ref{2lcycle matrix}).
    To obtain a cycle in $H_{[0,s]}^1$ the matrices that are inserted at these positions have to contain $1$s such that they also form such a submatrix. This means that
    \begin{align*}
        \left(\begin{smallmatrix}
            Q_{a_1^1}^{r_1}(a_1^2,b_1^2) & Q_{a_1^1}^{r_2}(a_1^2,b_2^2)&0&\cdots&0\\
            0& Q_{a_2^1}^{r_2}(a_2^2,b_2^2)&\hspace{4mm}Q_{a_2^1}^{r_3}(a_2^2,b_3^2)&\ddots&\vdots\\
            \vdots&&\ddots&\ddots&0\\
            0&&&Q_{a_{\ell-1}^1}^{r_{\ell-1}}(a_{\ell-1}^2,b_{\ell-1}^2)&Q_{a_{\ell-1}^1}^{r_{\ell}}(a_{\ell-1}^2,b_\ell^2)\\
            Q_{a_\ell^1}^{r_1}(a_\ell^2,b_1^2)&0&\cdots&0&Q_{a_\ell^1}^{r_\ell}(a_\ell^2,b_\ell^2)
        \end{smallmatrix}\right)= \left(\begin{smallmatrix}
        1&1&0&\cdots &0\\
        0&1&1&\ddots &\vdots\\
        \vdots & & \ddots&\ddots&0\\
        0&&& 1&1\\
        1&0&\cdots &0&1
    \end{smallmatrix}\right)
    \end{align*}
    for $a_1^2,\dots,a_\ell^2,b_1^2,\dots,b_\ell^2\in\{1,\dots,p\}$.
    By repeating this argument for $H_{[0,s]}^2,\dots,H_{[0,s]}^m$ we get that
    \begin{align*}
        \left(\begin{smallmatrix}
            Q_{a_1^{h-1}}^{r_1}(a_1^h,b_1^h) & Q_{a_1^{h-1}}^{r_2}(a_1^h,b_2^h)&0&\cdots&0\\
            0& Q_{a_2^{h-1}}^{r_2}(a_2^h,b_2^h)&\hspace{4mm}Q_{a_2^{h-1}}^{r_3}(a_2^h,b_3^h)&\ddots&\vdots\\
            \vdots&&\ddots&\ddots&0\\
            0&&&Q_{a_{\ell-1}^{h-1}}^{r_{\ell-1}}(a_{\ell-1}^h,b_{\ell-1}^h)&Q_{a_{\ell-1}^{h-1}}^{r_{\ell}}(a_{\ell-1}^h,b_\ell^h)\\
            Q_{a_\ell^{h-1}}^{r_1}(a_\ell^h,b_1^h)&0&\cdots&0&Q_{a_\ell^{h-1}}^{r_\ell}(a_\ell^h,b_\ell^h)
        \end{smallmatrix}\right)= \left(
    \begin{smallmatrix}
        1&1&0&\cdots &0\\
        0&1&1&\ddots &\vdots\\
        \vdots & & \ddots&\ddots&0\\
        0&&& 1&1\\
        1&0&\cdots &0&1
    \end{smallmatrix}\right)
    \end{align*}
    for $a_1^h,\dots,a_\ell^h,b_1^h,\dots,b_\ell^h\in\{1,\dots,p\}$ with $h=3,\dots,m+1$.

    These conditions yield the equations
    \begin{align}
        Q_{a^0_1}^{r_1}(a_1^1,b_1^1) &= Q_{a^0_j}^{r_j}(a_{j-1}^1,b_j^1) = Q_{a^0_j+\alpha_{j-1}}^{r_j}(a_j^1,b_j^1)= Q_{a_1^0+\sum\limits_{i=1}^{\ell-1}\alpha_i}^{r_1}(a_{\ell}^1,b_1^1)= 1&& \text{and} \label{Q1}\\
        Q_{a_{1}^{h-1}}^{r_1}(a_1^h,b_1^h) &= Q_{a_{j-1}^{h-1}}^{r_j}(a_{j-1}^h,b_j^h) = Q_{a_j^{h-1}}^{r_j}(a_j^h,b_j^h)= Q_{a_{\ell}^{h-1}}^{r_1}(a_{\ell}^h,b_1^h)=1 \label{Q2}
    \end{align}
    for $j=2,\dots,\ell$ and $ \ h = 2,\dots,m+1.$
    
    Since each Latin square has only one entry in each position, it holds that
    $a_\ell^h\neq a_1^h$ and $ a_{j-1}^h\neq a_{j}^h,j=2,\dots,\ell$ for all $h=0,\dots,m+1$.

    The following equations and calculations are all modulo $p$.
    From (\ref{Q1}) we get
    \begin{align}
        a^0_1 &= b_1^1-r_1(a_1^1-1) \label{Q11}\\
        a^0_j &= b_j^1-r_j(a_{j-1}^1-1) && j=2,\dots,\ell\label{Q12}\\
        a^0_j+\alpha_{j-1} &= b_j^1-r_j(a_j^1-1) && j=2,\dots,\ell\label{Q13}\\
        a^0_1 +\sum\limits_{i=1}^{\ell-1}\alpha_i &= b_1^1-r_1(a_{\ell}^1-1)\label{Q14}\\
        (\ref{Q14})-(\ref{Q11}): \sum\limits_{i=1}^{\ell-1}\alpha_i &= r_1(a_1^1-a_{\ell}^1)\label{formel 1}\\
        (\ref{Q13})-(\ref{Q12}): \alpha_{j-1}  &= r_j(a_{j-1}^1-a_j^1) && j=2,\dots,\ell\label{formel 2}\\
        (\ref{formel 1}),(\ref{formel 2})\implies 0&= r_1(a_{\ell}^1-a_{1}^1) +\sum\limits_{i=2}^{\ell} r_i(a_{i-1}^1-a_{i}^1)&& a_0^1:=a_{\ell}^1 \notag\\
        0&= \sum\limits_{i=1}^{\ell}r_i(a_{i-1}^1-a_{i}^1)\label{sum1} .
    \end{align}
    Similarly, one can deduce the following equations from (\ref{Q2}) for all $h=2,\dots,m+1$ and $ j=2,\dots,\ell$.
    \begin{align}
        a_1^{h-1} &= b_1^h-r_1(a_1^h-1) \label{Q21}\\
        a_{j-1}^{h-1} &= b_j^h-r_j(a_{j-1}^h-1)  \label{Q22}\\
        a_{j}^{h-1} &= b_j^h-r_j(a_{j}^h-1) \label{Q23}\\
        a_{\ell}^{h-1} &= b_1^h-r_1(a_{\ell}^h-1) \label{Q24}\\
        (\ref{Q24})-(\ref{Q21}): a_{\ell}^{h-1}-a_1^{h-1} &= r_1(a_1^h-a_{\ell}^h) \label{formel 3}\\
        (\ref{Q23})-(\ref{Q22}): a_j^{h-1}-a_{j-1}^{h-1}&= r_j(a_{j-1}^h-a_j^h) \label{formel 4}\\
        (\ref{formel 3}),(\ref{formel 4})\implies 0 &= \sum\limits_{i=1}^{\ell} r_i(a_{i-1}^h-a_i^h)&& a_0^h:=a_\ell^h \label{sum2}\\
        \intertext{It follows for $g=h,\dots,m+1$ and $h=1,\dots,m+1$ and $j= 1,\dots,\ell$ that}
        (\ref{formel 3}),(\ref{formel 4})\implies a_{j-1}^{h}- a_j^{h} &=(-r_j)^{g-h}(a_{j-1}^{g}- a_j^{g}) \label{formel 5}.
    \end{align}
    By combining these equations we get
    \begin{align}
        (\ref{sum1}),(\ref{sum2}),(\ref{formel 5})\implies 0&= \sum\limits_{i=1}^{\ell} r_i(a_{i-1}^{m+1}-a_i^{m+1})\notag\\
         0&=\sum\limits_{i=1}^\ell r_i(a_{i-1}^m-a_i^m) = - \sum\limits_{i=1}^{\ell} r_i^2(a_{i-1}^{m+1}-a_i^{m+1})\notag\\
         \vdots& \qquad \qquad \vdots \qquad \qquad \vdots\notag\\
          0&=\sum\limits_{i=1}^\ell r_i(a_{i-1}^{m-\ell+3}-a_i^{m-\ell+3}) = (-1)^{\ell-2} \sum\limits_{i=1}^{\ell} r_i^{\ell-1}(a_{i-1}^{m+1}-a_i^{m+1})\notag \\
          \implies 0&=\sum\limits_{i=1}^{\ell} r_i^{L}(a_{i-1}^{m+1}-a_i^{m+1}) \notag
    \end{align}
    for $L=1,\dots,\ell-1$ where $m-\ell+3\geq 1$ because $2\ell\leq 2(m+2)$.
    By renaming the variables through $c_j := a_j^{m+1}$ for $ j=0,\dots,\ell$ we get the system of equations
    \begin{align}
        \begin{pmatrix}
            1& 1& 1& \cdots &1\\
            r_1 & r_2& r_3& \cdots &r_{\ell}\\
            r_1^2 & r_2^2& r_3^2& \cdots& r_{\ell}^2\\
            \vdots &\vdots &\vdots &\ddots &\vdots\\
            r_1^{\ell-1} & r_2^{\ell-1}& r_3^{\ell-1}& \cdots& r_{\ell}^{\ell-1}\\
        \end{pmatrix} 
        \begin{pmatrix}
            c_{\ell}-c_1\\ c_1-c_2\\ c_2-c_3\\ \vdots \\ c_{\ell-1}-c_{\ell}
        \end{pmatrix} =V(r_1,r_2,\dots,r_{\ell})^T\cdot \begin{pmatrix}
            c_{\ell}-c_1\\ c_1-c_2\\ c_2-c_3\\ \vdots \\ c_{\ell-1}-c_{\ell}
        \end{pmatrix} =0\label{Gleichungssystem matrix}
    \end{align}
    where $V(r_1,r_2,\dots,r_{\ell})$ is a Vandermonde matrix. If $r_i\neq r_j$ for all $i,j= 1,\dots,\ell,i\neq j$, then $\det V(r_1,\dots,r_{\ell})= \prod_{1\leq i<j\leq \ell}(r_j-r_i)\neq 0$ 
    which implies $\left(\begin{smallmatrix} c_{\ell}-c_1\\ c_1-c_2\\ c_2-c_3\\ \vdots \\ c_{\ell-1}-c_{\ell} 
    \end{smallmatrix}\right)=0$.
        This cannot be true since $c_{i-1}\neq c_i$ for all $i=1,\dots,\ell$.

To understand what is happening if we assume that the statement of the lemma is not true, i.e. assuming that there exists at least one $i\in\{1,\hdots,\ell\}$ such that $r_i\neq r_j$ for all $j\in\{1,\dots,\ell\}\setminus\{i\}$, consider first the example $\ell=4$, i.e.,
   \begin{align}
        \begin{pmatrix}
            1& 1& 1& 1\\
            r_1 & r_2& r_3& r_4\\
            r_1^2 & r_2^2& r_3^2& r_4^2\\
            r_1^3 & r_2^3& r_3^3& r_4^3
        \end{pmatrix} 
        \begin{pmatrix}
            c_{4}-c_1\\ c_1-c_2\\ c_2-c_3 \\c_3-c_4
        \end{pmatrix} =0
    \end{align} with $r_1\neq r_2\neq r_3\neq r_1$ and $r_2=r_4$ , which can be reduced to \begin{align}
        \begin{pmatrix}
            1& 1& 1 \\
            r_1 & r_2& r_3\\
            r_1^2&r_2^2&r_3^2
        \end{pmatrix} 
        \begin{pmatrix}
            c_4-c_1\\ c_1-c_2+(c_3-c_4)\\ c_2-c_3
        \end{pmatrix} =0
    \end{align}
implying $c_1=c_4$ and $c_2=c_3$, which is a contradiction.

In general, if there exists at least one $i\in\{1,\hdots,\ell\}$ such that $r_i\neq r_j$ for all $j\in\{1,\dots,\ell\}\setminus\{i\}$,
then \Cref{Gleichungssystem matrix} can be reduced to a smaller homogeneous system of equations with a coefficient matrix where we only keep a subset of the columns of the Vandermonde matrix corresponding to a subset of $\{r_1,\hdots,r_{\ell}\}$ with pairwise distinct elements. If we have $N$ such pairwise distinct elements, we know $N\geq 2$ and the first $N$ rows of this reduced coefficient matrix form a Vandermonde matrix with nonzero determinant, where the column of this reduced matrix corresponding to $r_i$ is still multiplied by $c_{i-1}-c_i$.   
    Hence, it follows that $c_{i-1}=c_i$, which is a contradiction. Therefore, for each $i\in\{1,\dots,\ell\}$, there is $j\in\{1,\dots,\ell\}\setminus\{i\}$ such that $r_i=r_j$, which proves the lemma.
\end{proof}

    \begin{theorem}\label{Theo girth 8, no 10c}
        The parity-check matrix $H_{[0,s]}^1$ has no cycles of length $6$ and there are no cycles of length $10$ in $H_{[0,s]}^3$. More generally, every parity check matrix $H_{[0,s]}^m$ with $m\geq 1$ has girth at least $8$ and there are no $10$ cycles in $H_{[0,s]}^m$ for $m\geq 3$.
    \end{theorem}
    \begin{proof}
        It follows from \Cref{theo H0 girth 6} that $H_{[0,s]}^m$ has girth at least $6$ for all $m\geq 0$.
        
        If there was a cycle of length $6$ in $H_{[0,s]}^m$ for $m\geq 1$, then it follows from \Cref{lemma r pw gleich} that there would be $r_1,r_2,r_3$ with $r_1\neq r_2\neq r_3\neq r_1$ such that for
        or each $i\in\{1,2,3\}$, there is $j\in\{1,2,3\}\setminus\{i\}$ such that $r_i=r_j$, which is not possible. So, there is no cycle of length $6$.
    
        If there was a cycle of length $10$ in $H_{[0,s]}^m$ for $m\geq 3$, then it follows from \Cref{lemma r pw gleich} that there would be $r_1,r_2,r_3,r_4,r_5$ with $r_1\neq r_2\neq r_3\neq r_4\neq r_5\neq r_1$ such that 
        for
        or each $i\in\{1,\hdots,5\}$, there is $j\in\{1,\hdots,5\}\setminus\{i\}$ such that $r_i=r_j$.
        Thus, $r_1$ has to be equal to $r_3$ or $r_4$, but cannot be equal to both. For reasons of symmetry, we can assume that $r_1=r_3$. Then, since $r_4$ is different from $r_5$ and $r_3$, and hence also different from $r_1$, one obtains that $r_4$ must be equal to $r_2$. Since $r_5$ is different from $r_4$ and $r_1$, it cannot be equal to $r_2=r_4$ or $r_3=r_1$. Thus, it is not possible for each of the $r_i,i=1,2,3,4,5$ to be equal to one of the others and there is no cycle of length $10$.
    \end{proof}
    
\subsection{Construction of LDPC convolutional codes with girth 10 and 12}\label{subsec girth 10,12}

    We now further modify the construction of Subsection \ref{subsec girth 8} to eliminate $8$-cycles. This results in codes with girths $10$ and $12$.
    \begin{definition}
        Let $\Tilde{H}_{[0,s]}^m$ be the matrix that is constructed from $H_{[0,s]}^m$ by replacing each entry with a matrix of size $p\times p$ as follows: each $0$ is replaced with $0_{p\times p}$, each $1$ in $I_p$ is replaced with $I_p$ and each $1$ in $Q_i^r(a,b)$ is replaced with $Q_{rab\mod p}^r$.
        The convolutional code derived from $\Tilde{H}_{[0,s]}^m$ is called $\Tilde{\mathcal{C}}^m$ and is a binary, time-varying $(2p^{m+2},p^{m+2})$ convolutional code with period $p-1$.
    \end{definition}
    \begin{theorem}\label{Theo girth 10,12}
        The convolutional code $\Tilde{\mathcal{C}}^m$ has girth at least $10$ for all $m\geq 2$ and girth at least $12$ for all $m\geq 3$.
    \end{theorem}
    \begin{proof}
        Consider $\Tilde{H}_{[0,s]}^m$ with $m\geq 2$. Since $H_{[0,s]}^m$ has no cycles of length $4$ or $6$, the girth of $\Tilde{H}_{[0,s]}^m$ is at least $8$. Assume that there is a cycle of length $8$ in $\Tilde{H}_{[0,s]}^m$. Then, there is also a cycle in $H_{[0,s]}^m$ corresponding to the same $r_1,r_2,r_3,r_4\in\{1,\dots,p-1\}$ with $r_1\neq r_2\neq r_3\neq r_4\neq r_1$. We know from \Cref{lemma r pw gleich} that none of these $r_j,j=1,2,3,4$ can be distinct from all of the others and hence, $r_1=r_3$ and $r_2=r_4$. The following equations and calculations are again all modulo $p$. 
        We get from \Cref{Gleichungssystem matrix} that
        \begin{align}
            \begin{pmatrix}
                1&1\\
                r_1 &r_2
            \end{pmatrix}\begin{pmatrix}
                c_4-c_1+c_2-c_3 \\ c_1-c_2+c_3-c_4
            \end{pmatrix}=0\notag\\
            \implies c_1-c_2+c_3-c_4=0\label{girth10 c}
        \end{align}
        and from \Cref{Q21}-(\ref{Q24}) it follows for $h=m+1$ and $d_j:=b_j^{m+1},j=1,\dots,4$ that\\
        \begin{minipage}{0.44\textwidth}
            \begin{align}
            a_1^m &= d_1-r_1(c_1-1)\label{g10 1}\\
            a_1^m &= d_2-r_2(c_1-1)\label{g10 2}\\
            a_2^m &= d_2-r_2(c_2-1)\label{g10 3}\\
            a_2^m &= d_3-r_3(c_2-1)\label{g10 4}
            \end{align}
        \end{minipage}
        \begin{minipage}{0.44\textwidth}
            \begin{align}
            a_3^m &= d_3-r_3(c_3-1)\label{g10 5}\\
            a_3^m &= d_4-r_4(c_3-1)\label{g10 6}\\
            a_4^m &= d_4-r_4(c_4-1)\label{g10 7}\\
            a_4^m &= d_1-r_1(c_4-1)\label{g10 8}
            \end{align}
        \end{minipage}
        \begin{align}
            (\ref{g10 2})-(\ref{g10 1})+d_1-d_2: d_1-d_2 &= (c_1-1)(r_1-r_2)\label{g10 9}\\
            (\ref{g10 4})-(\ref{g10 3})+d_2-d_3: d_2-d_3 &= (c_2-1)(r_2-r_3)\label{g10 10}\\
            (\ref{g10 6})-(\ref{g10 5})+d_3-d_4: d_3-d_4 &= (c_3-1)(r_3-r_4)\label{g10 11}\\
            (\ref{g10 8})-(\ref{g10 7})+d_4-d_1: d_4-d_1 &= (c_4-1)(r_4-r_1)\label{g10 12}
        \end{align}
        In addition, the following has to hold for $e_1,\dots,e_4,f_1,\dots,f_4\in\{1,\dots,p\}$.
        \begin{align*}
            Q_{r_1 c_1 d_1}^{r_1}(e_1,f_1) &= Q_{r_2 c_1 d_2}^{r_2}(e_1,f_2) = Q_{r_2 c_2 d_2}^{r_2}(e_2,f_2) = Q_{r_3 c_2 d_3}^{r_3}(e_2,f_3) = 1,\\
            Q_{r_3 c_3 d_3}^{r_3}(e_3,f_3) &= Q_{r_4 c_3 d_4}^{r_4}(e_3,f_4) = Q_{r_4 c_4 d_4}^{r_4}(e_4,f_4) = Q_{r_1 c_4 d_1}^{r_1}(e_4,f_1) = 1,
        \end{align*}
        
        which implies\\
        \begin{minipage}{0.44\textwidth}
            \begin{align}
            r_1c_1d_1&= f_1-r_1(e_1-1)\label{girth10 1}\\
            r_2c_1d_2&= f_2-r_2(e_1-1)\label{girth10 2}\\
            r_2c_2d_2&= f_2-r_2(e_2-1)\label{girth10 3}\\
            r_3c_2d_3&= f_3-r_3(e_2-1)\label{girth10 4}
            \end{align}
        \end{minipage}
        \begin{minipage}{0.44\textwidth}
            \begin{align}
            r_3c_3d_3&= f_3-r_3(e_3-1)\label{girth10 5}\\
            r_4c_3d_4&= f_4-r_4(e_3-1)\label{girth10 6}\\
            r_4c_4d_4&= f_4-r_4(e_4-1)\label{girth10 7}\\
            r_1c_4d_1&= f_1-r_1(e_4-1)\label{girth10 8}
            \end{align}
        \end{minipage}

        \begin{align}
            \frac{1}{r_1}\cdot ((\ref{girth10 1})-(\ref{girth10 8})): e_4-e_1=& d_1(c_1-c_4)  \label{girth10 9}\\
            \frac{1}{r_2}\cdot ((\ref{girth10 3})-(\ref{girth10 2})): e_1-e_2 =& d_2(c_2-c_1)\label{girth10 10}\\
            \frac{1}{r_3}\cdot ((\ref{girth10 5})-(\ref{girth10 4})): e_2-e_3=& d_3(c_3-c_2) \label{girth10 11}\\
            \frac{1}{r_4}\cdot ((\ref{girth10 7})-(\ref{girth10 6})): e_3-e_4 =& d_4(c_4-c_3)\label{girth10 12}
        \end{align}
        \begin{align}
            (\ref{girth10 9})+(\ref{girth10 10})+(\ref{girth10 11})+(\ref{girth10 12}): 0=&c_1(d_1-d_2)+c_2(d_2-d_3)+c_3(d_3-d_4)+c_4(d_4-d_1) \label{girth10 f1}\\
            (\ref{g10 9}),(\ref{g10 10}),(\ref{g10 11}),(\ref{g10 12}) \rightarrow (\ref{girth10 f1}):  0=&c_1^2(r_1-r_2)+c_2^2(r_2-r_3)+c_3^2(r_3-r_4)+c_4^2(r_4-r_1) \notag\\
            & -(c_1(r_1-r_2)+c_2(r_2-r_3)+c_3(r_3-r_4)+c_4(r_4-r_1)\notag\\
            =&(r_1-r_2)(c_1^2-c_2^2+c_3^2-c_4^2-(c_1-c_2+c_3-c_4))\notag\\
            (\ref{girth10 c}) \implies 0=&(r_1-r_2)(c_1^2-c_2^2+c_3^2-(c_1-c_2+c_3)^2)\notag\\
            =&2(r_1-r_2)(c_2-c_3)(c_1-c_2)\notag .
        \end{align}
As $r_1\neq r_2$, this implies either $c_2=c_3$ or $c_1=c_2$, both of which is not possible and hence, there are no cycles of length $8$ in $\Tilde{H}_{[0,s]}^m$ for $m\geq 2$. Since, according to Theorem \ref{Theo girth 8, no 10c}, there are no cycles of length $10$ in $H_{[0,s]}^m$ for $m\geq 3$, there are also no cycles of length $10$ in $\Tilde{H}_{[0,s]}^m$ for $m\geq 3$.
    \end{proof}
    
\subsection{Density and distance properties of the constructed codes}\label{subsec other prop}

In this subsection, we evaluate the LDPC codes constructed in the previous subsections with respect to other relevant properties besides the girth of the associated Tanner graph, namely with respect to the density of their parity-check matrices and their free distance and column distances.
\begin{remark}
    By easy calculations, one obtains that the density of $H_{[0,s]}^m$ is equal to $$\frac{(s+1)p^{m+1}(\mu+2)}{2(s+1)p^{2(m+1)}(\mu+s+1)}=\frac{\mu+2}{2p^{m+1}(\mu+s+1)}=o(n).$$ If $\mu$ assumes the maximum possible value $p-2$, this density is equal to $\frac{1}{2p^m(p+s-1)}$.
    The density of $\Tilde{H}_{[0,s]}^m$ is the same as the density of $H_{[0,s]}^{m+1}$.
\end{remark}

    \begin{theorem}\label{Theo dist}
        The codes $\mathcal{C}^m$ and $\Tilde{\mathcal{C}}^m$ with sliding parity-check matrices $H_{[0,s]}^m$ and $\Tilde{H}_{[0,s]}^{m}$, respectively, have the following distance properties for all $m\in \mathbb{N}_0$:
        \begin{enumerate}[label=(\roman*)]
            \item $d_j^c=\min\{j,\mu\}+2$
            \item $d_{\text{free}}=\mu+2$.
        \end{enumerate}
    \end{theorem}
  
    \begin{proof}
        We prove the statement only for $\mathcal{C}^m$ as it is almost identical for $\Tilde{\mathcal{C}}^{m}$.
        
        Due to Theorem \ref{Th djc Hjc time}, the column distance $d^c_j$ is equal to $d$ if, for all $t\in\mathbb N_0$, none of the first $2p^{m+1}$ columns of $H^c_j(t)$ is contained in the span of any other $d-2$ columns and if there exists $t\in\mathbb N_0$ such that one of the first $2p^{m+1}$ columns of $H_j^c(t)$ is in the span of some other $d-1$ columns of that matrix.
    
        Since there are no $4$-cycles in $H_{[0,s]}^m$ for all $m\in \mathbb{N}_0$, two different columns have at most one $1$ in common.
            For all $t\in\mathbb{N}_0$, each of the first $p^{m+1}$ columns of $H_j^c(t)$ contains $\min\{j+1,\mu+1\}$ entries equal to $1$.
            Therefore, one has to add at least $\min\{j+1,\mu+1\}$ other columns in order to achieve that the sum of the columns is zero.
        Each of the first $2p^{m+1}$ columns that is not part of the first $p^{m+1}$ columns contains exactly one entry equal to $1$, which can only be turned into a $0$ by adding one of the first $p^{m+1}$ columns. The sum of two such vectors contains at least $\min\{j,\mu\}$ entries equal to $1$, so there are at least $\min\{j,\mu\}$ other columns needed so that the sum of them is zero. 

        For each of the first $2p^{m+1}$ columns of $H_j^c(t)$, there exists exactly one other column in the first $2p^{m+1}$ which has a $1$ at the same position. One of these columns contains exactly one $1$ and the other $\min\{j+1,\mu+1\}$ many $1$s. The sum of two such columns contains $\min\{j,\mu\}$ nonzero entries. We can now choose columns that contain exactly one $1$ of the identity matrix and add them to obtain zero. We find that each of the first $2p^{m+1}$ columns is contained in the span of $\min\{j+1,\mu+1\}$ other columns. Hence, $d_j^c=\min\{j+1,\mu+1\}+1$.
    
        Since $d_\text{free}\geq d_j^c$ for all $j\in\mathbb{N}_0$, it holds that $d_\text{free}\geq \mu+2$. Similarly to the calculation of $d_j^c$, one obtains that the free distance cannot be larger than $\mu+2$, which proves the theorem.
    \end{proof}

It is important to note that when $p$ is increasing, the density of the parity-check matrix is decreasing and the free and column distances are increasing. Hence, codes with larger values of $p$, which of course also implies larger parity-check matrices, have a better performance in terms of efficiency and error-correction.

\section{Construction of time-invariant binary LDPC convolutional codes with large girth}

In this section, we use construction ideas similar to those in the previous section to obtain time-invariant LDPC convolutional codes with large girth.

Since the parity-check matrices $H_{[0,s]}^m$ and $\Tilde{H}_{[0,s]}^m$ from the previous section always have period $p-1$, we can obtain a time-invariant convolutional code with sliding parity-check matrix $\hat{H}$ in the following way:
\begin{align*}
    \hat H_0 &:=\begin{pmatrix}
        H_0(0)\\
        H_1(0) & H_0(1)\\
        \vdots & \vdots & \ddots \\
        H_{p-2}(0) & H_{p-3}(1) & \dots & H_0(p-2)
    \end{pmatrix},\qquad
    \hat H_1 := \begin{pmatrix}
        0&H_{p-2}(1) & \dots &H_1(p-2)\\
        &0&\ddots &\vdots\\
        &&\ddots&H_{p-2}(p-2)\\
        &&&0
    \end{pmatrix},\\
    \hat H &:= \begin{pmatrix}
        \hat H_0\\
        \hat H_1 &\hat H_0\\
        &\ddots&\ddots\\
        &&\hat H_1&\hat H_0\\
        &&&\hat H_1
    \end{pmatrix} = \begin{pmatrix}
        H_0(0)\\
            H_1(0)& H_0(1)\\
            \vdots& H_1(1)& \ddots &\\
            H_\mu (0) &\vdots &\ddots &H_0(s)\\
            & H_\mu(1) & & H_1(s)\\
            &&\ddots& \vdots\\
            &&& H_\mu(s)
    \end{pmatrix}.
\end{align*}
The code $\hat{\mathcal{C}}$ constructed from this parity-check matrix $\hat H$ has the same girth as the time-varying code from which it is derived.
These codes have the disadvantage that the memory is always equal to $1$.

We are now going to construct a new code which is time invariant, has no cycles of length $4$ and can have memory larger than $1$.
For this we define the modified incidence matrices for the Latin squares $L_1,\dots,L_{p-1}$ from \Cref{subsec LS} in the following way:
\begin{align}
    \Tilde{Q}_i^r(a,b) := \left\{\begin{array}{ll}1 & \text{if } L_r(a+1,b+1)=i \\ 0 & \, \text{otherwise} \\ \end{array} \right. .\notag
\end{align}
Note that $\Tilde{Q}_i^r\in \mathbb{F}_2^{(p-1)\times(p-1)}$ is obtained from the incidence matrix $Q_i^r$ by deleting the first column and the first row.
We now define the sliding parity-check matrix 
\begin{align}
    H^\prime := \begin{pmatrix}
        \Tilde{Q}_1^1&I_{p-1}\\
        \Tilde{Q}_1^2&0 & \Tilde{Q}_1^1&I_{p-1}\\
        \vdots&\vdots & \Tilde{Q}_1^2&0&\ddots\\
        \Tilde{Q}_1^\mu&0 & \vdots&\vdots &\ddots &\Tilde{Q}_1^1&I_{p-1}\\
        \Tilde{Q}_1^{\mu+1}&0 & \Tilde{Q}_1^{\mu}&0 & &\Tilde{Q}_1^2&0\\
        && \Tilde{Q}_1^{\mu+1}&0&\ddots &\vdots&\vdots \\
        &&&&\ddots&\Tilde{Q}_1^\mu&0\\
        &&&&&\Tilde{Q}_1^{\mu+1}&0
    \end{pmatrix}\notag
\end{align}
with $\mu\leq p-2$. 

\begin{theorem}
    The binary, time-invariant $(2(p-1),p-1)$ convolutional code $\mathcal{C}^\prime$ derived from $H^\prime$ has girth at least $6$.
\end{theorem}
\begin{proof}
    Assume that there exists a cycle of length $4$. Then $H^\prime$ contains the submatrix $\left(\begin{smallmatrix}
        1&1\\1&1
    \end{smallmatrix}\right)$. Similarly to the proof of \Cref{theo H0 girth 6} we get that
    \begin{align}
        \Tilde{Q}_{1}^{r_1}(a_1,b_1) = \Tilde{Q}_{1}^{r_2}(a_1,b_2) = \Tilde{Q}_{1}^{r_1+\alpha}(a_2,b_1) = \Tilde{Q}_{1}^{r_2+\alpha}(a_2,b_2) = 1\notag
    \end{align}
    for $r_1,r_2,a_1,a_2,b_1,b_2,\alpha \in \{1,\dots,p-1\}$ with $r_1\neq r_2$. We know that $a_1\neq a_2$ because two orthogonal Latin squares have a $1$ at the same position exactly once and, by definition, this position is $(1,1)$, which is not part of the modified incidence matrices.
    The following equations and calculations are again all modulo $p$. We get
    \begin{align}
        1&= b_1+1-r_1a_1\label{H prime 1}\\
        1&= b_2+1-r_2a_1\label{H prime 2}\\
        1&= b_1+1-(r_1+\alpha)a_2\label{H prime 3}\\
        1&= b_2+1-(r_2+\alpha)a_2 \label{H prime 4}\\
        (\ref{H prime 1})-(\ref{H prime 2}): 0&= b_1-b_2 +a_1(r_2-r_1)\label{H prime 5}\\
        (\ref{H prime 3})-(\ref{H prime 4}): 0&= b_1-b_2 +a_2(r_2-r_1)\label{H prime 6}\\
        (\ref{H prime 5})-(\ref{H prime 6}): 0&= (r_2-r_1)(a_1-a_2)\notag 
    \end{align}
    which proves that there cannot be a cycle of length $4$, as $a_1\neq a_2$ and $r_1\neq r_2$.
\end{proof}

\begin{remark}
    We conjecture that $\mathcal{C}^\prime$ has girth exactly $6$ for $p\geq 5$ and $\mu\geq 2$. For $p=5$ the $1$s at the positions $\Tilde{Q}_1^2(1,2), \Tilde{Q}_1^1(1,1), \Tilde{Q}_1^3(2,1), \Tilde{Q}_1^2(2,4), \Tilde{Q}_1^1(4,4)$ and $\Tilde{Q}_1^3(4,2)$ form a cycle of length $6$.
\end{remark}

\section{An LDPC block code construction without 4 cycles and 6 cycles}\label{sec blockcode}

In this section, we consider the construction of LDPC block codes from Latin squares presented in \cite{BlockCodes}. Although the codes constructed there do not contain cycles of length $4$, there are still some cycles of length $6$. We are now going to use the methods we used for LDPC convolutional codes in the previous sections to modify the construction from \cite{BlockCodes} to achieve girth $8$.
In the following, we first briefly describe the construction from \cite{BlockCodes}.

\begin{definition}
    A \textbf{one-configuration} is an ordered pair $(V,\mathcal{B})$ where $V$ consists of $v$ elements and $\mathcal{B}$ contains subsets of size $t$ of $V$ called blocks and each pair of elements of $V$ appears together in at most one block.
\end{definition}

The construction in \cite{BlockCodes} depends on the Latin square of order $2m+1$ defined by
\begin{align*}
    L(i,j):=\frac{i+j}{2} \mod (2m+1)
\end{align*}
for a positive integer $m$.
From this Latin square, the authors construct a one-configuration $(V,\mathcal{B})$ where $V:=\{1,\dots,2m+1\}\times \{1,2,3\}$, i.e. $|V|=3(2m+1)$, and $\mathcal{B}$ contains all subsets of the form 
$\{(i,a),(j,a),(L(i,j),a+1\mod 3)\}$ with $a=1,2,3$ and $1\leq i<j\leq 2m+1$, i.e. $|\mathcal{B}|=3\cdot \binom{2m+1}{2}=3m(2m+1)$.
We denote the elements of $V$ by $v_1,\dots,v_{6m+3}$ and the elements of $\mathcal{B}$ with $B_1,\dots,B_{6m^2+3m}$.
Then, the authors use the incidence matrix 
\begin{align*}
    H(i,j)=1 \iff v_i\in B_j
\end{align*}
of this one-configuration as a parity-check matrix of an LDPC block code. It is easy to see that such a matrix is free of $4$-cycles.

By rearranging the columns of $H$, and forming blocks of rows and columns of size $3$, the matrix contains submatrices of size $3\times 3$ which are either $I$, $0$ or $P^2$ where $P^i$ denotes the permutation matrix that results from the identity matrix by shifting the columns $i$ positions to the right.
In each block column there are two identity matrices, say in block rows $i$ and $j$, then there is $P^2$ in block row $L(i,j)$ and the rest is $0$.

Now we can rearrange and group the block columns into $m$ submatrices $M_1,\dots,M_m$ such that $M_\ell$ contains the block columns where the identity matrices are in block rows $i$ and $j$ with $j=i+2\ell\mod (2m+1)$ for $i=1,\dots,2m+1$. This implies that the matrix $P^2$ is in block row $$L(i,j)=\frac{i+j}{2}\mod(2m+1)=\frac{2i+2\ell}{2}\mod(2m+1)=i+\ell\mod(2m+1).$$
So each of the submatrices $M_1,\dots,M_m$ contains exactly one $P^2$ per block row.

For example, for $m=2$ we get
\begin{align*}
    M_1=\begin{pmatrix}
        I&&&I&P^2\\
        P^2&I&&&I\\
        I&P^2&I\\
        &I&P^2&I\\
        &&I&P^2&I
    \end{pmatrix} \quad \text{and} \quad
    M_2 = \begin{pmatrix}
        I&I& &P^2\\
        &I&I& &P^2\\
        P^2&&I&I\\
        &P^2&&I&I\\
        I&&P^2&&I
    \end{pmatrix}
\end{align*}
where the block columns of size $3$ are sorted by increasing $i$ and $H=\begin{pmatrix}
    M_1 &\mid &M_2
\end{pmatrix}$.

\begin{lemma}\cite{Fan2001}
    The permutation matrices $P^{i_1},P^{i_2},P^{i_3},P^{i_4},P^{i_5}$ and $P^{i_6}$ of order $D$ with $i_1,\dots,i_6\in\{0,\dots,D-1\}$ that are arranged in a block cycle as $\left(\begin{smallmatrix}
        P^{i_1}&P^{i_2}\\
        &P^{i_3}&P^{i_4}\\
        P^{i_6}&&P^{i_5}
    \end{smallmatrix}\right)$ contain a cycle of length $6$ if and only if their Fan sum $i_1-i_2+i_3-i_4+i_5-i_6 \mod D$ is zero.
\end{lemma}

Because of that, in \cite{BlockCodes}, the block-cycles of length $6$ in $H$ were classified into the following four classes:

\begin{itemize}
    \item [] Class A: only identity matrices are part of the block-cycle
    \item [] Class B: two $P^2$ matrices in the same block row are in the block-cycle
    \item [] Class C: two $P^2$ matrices in different block rows are in the block-cycle
    \item [] Class D: three $P^2$ matrices take part in the block-cycle.
\end{itemize}
\begin{minipage}{0.23\textwidth}
    \begin{align*}
        \begin{pmatrix}
        I&I\\
        &I&I\\
        I&&I
    \end{pmatrix}\\
    \text{ Class A }
    \end{align*}
\end{minipage}
\begin{minipage}{0.23\textwidth}
    \begin{align*}
        \begin{pmatrix}
        I&I\\
        &P^2&P^2\\
        I&&I
    \end{pmatrix}\\
    \text{Class B }\quad  
    \end{align*}
\end{minipage}
\begin{minipage}{0.23\textwidth}
    \begin{align*}
        \begin{pmatrix}
        I&P^2\\
        &I&I\\
        I&&P^2\\
    \end{pmatrix}\\
    \text{Class C }\quad
    \end{align*}
\end{minipage}
\begin{minipage}{0.23\textwidth}
    \begin{align*}
        \begin{pmatrix}
        I&P^2\\
        &I&P^2\\
        P^2&&I
    \end{pmatrix}\\
    \text{Class D }\quad
    \end{align*}
\end{minipage}\\

We are now going to eliminate the cycles of length $6$ in several consecutive steps. 
\begin{itemize}
    \item[] Step 1: 
    \begin{itemize}
        \item[] To eliminate the $6$ cycles of Class D, the authors of \cite{BlockCodes} proposed to replace each $3\times 3$ matrix with an analogous $5\times 5$ matrix.
    \end{itemize}
    \item [] Now we want to replace each entry in $H$ with some matrix to get a parity-check matrix also without $6$ cycles of Classes A, B and C.
    \item[] Step 2:
    \begin{itemize}
        \item[] Since Class C contains two matrices $P^2$ in different block rows, we are going to replace each $1$ in an identity matrix of the initial parity-check matrix with $I_{2m+1}=P^0$, each $0$ with $0_{(2m+1)\times (2m+1)}$ and each $1$ in $P^2$ in block row $i$ of the initial parity-check matrix with $P^i$ of order $2m+1$. Then, the Fan sum of such newly created block cycles is $i-j$ or $j-i$ if the matrices $P^2$ are in block rows $i$ and $j$. This sum can never be zero modulo $2m+1$ as $1\leq i<j\leq 2m+1$,        
       and hence we do not have cycles of Class C in this new construction.
    \end{itemize}
    \item[] Step 3:
    \begin{itemize}
        \item[] 
        For the cycles of Class B we notice that two matrices $P^2$ that are in the same block row are in different submatrices $M_{\ell_1}$ and $M_{\ell_2}$, i.e. $\ell_1,\ell_2\in\{1,\hdots,m\}$, $\ell_1\neq \ell_2$. So we replace each $1$ that is in an identity matrix of the initial parity-check matrix with $I_m=P^0$, each $0$ with $0_{m\times m}$ and each $1$ in $P^2$ in submatrix $M_\ell$ of the initial parity-check matrix with $P^\ell$ of order $m$. So, the Fan sum is equal to $\ell_1-\ell_2$ or $\ell_2-\ell_1$ which cannot be $0$ modulo $m$. So we eliminated the cycles of Class B.
    \end{itemize}
    \item[] Step 4:
    \begin{itemize}
        \item[] Since the cycles in Class A do not contain $P^2$ matrices, it does not matter with what we replace the ones there. So we replace the ones in $P^2$ of the initial parity-check matrix with $I_5$ and all zeros with $0_{5\times 5}$. Each block column of the initial parity-check matrix contains two identity matrices in block rows $i$ and $j$ with $i<j$. We replace each $1$ in the identity matrix in block row $i$ with $P$ of order $5$ and the $1$s in block row $j$ with $P^2$ of order 5. 
        Thus, each block column (consisting of $5$ columns) of all newly created block cycles contains a $P$ and a $P^2$. That means that one of them contributes a positive number to the Fan sum and the other a negative number. Hence, every block column contributes either $1$ or $-1$ to the Fan sum. 
        Then the Fan sum is either $3,1,-1$ or $-3$ and therefore never $0$ modulo $5$.
    \end{itemize}
\end{itemize}

By replacing each entry with a matrix in these four steps, we remove all six cycles from the original parity-check matrix and no new cycles are created. Hence, we obtain a parity-check matrix with girth at least $8$. 

The original parity-check matrix is of size $3(2m+1)\times 3m(2m+1)$. After step 1 the size of the parity-check matrix is $5(2m+1)\times 5m(2m+1)$. The size is then multiplied by $2m+1$ in step 2, by $m$ in step 3 and by $5$ in step 4 resulting in a parity check matrix of size $25m(2m+1)^2\times 25m^2(2m+1)^2$.

\section{Conclusion}
In this paper, we present structured constructions of LDPC codes with large girths using Latin squares. These constructions cover time-variant and time-invariant convolutional codes as well as block codes. 
Our analysis shows that the constructed LDPC convolutional codes perform better when the size of the underlying Latin squares is larger.
An interesting problem for future research would be to investigate whether different combinatorial objects can be used to obtain further constructions of LDPC convolutional codes with large girth.

\section*{Acknowledgements}
This work has been supported by the German research foundation, project number 513811367.

\bibliographystyle{abbrv} 
\bibliography{bib}

\end{document}